\documentclass[lettersize,journal]{IEEEtran}
\usepackage{amsmath,amsfonts}
\usepackage{amssymb}  
\usepackage{amsthm}
\usepackage{algorithmic}
\usepackage{algorithm}
\usepackage{array}
\usepackage[caption=false,font=normalsize,labelfont=sf,textfont=sf]{subfig}
\usepackage{textcomp}
\usepackage{stfloats}
\usepackage{url}
\usepackage{verbatim}
\usepackage{graphicx}

\usepackage{cite}
\newtheorem{theorem}{Theorem}

\begin{document}
    \title{Tripartite System Dynamics of Cooperation Ecological Resources and Adaptive Regulation}
\author{Linjie Liu, Yike Li, Shijia Hua
\thanks{L. Liu, Y. Li, S. Hua are with College of Science, Northwest A \& F University, Yangling, 712100, China.}
\thanks{Corresponding authors: Shijia Hua (email: sjhua@nwafu.edu.cn).}}

\maketitle

\begin{abstract}
To address commons dilemmas under costly social cooperation, this paper develops a tripartite co-evolutionary system that integrates social behavior, ecological-resource dynamics, and institutional regulation. Institutional intervention is modeled as an adaptive feedback mechanism whose intensity responds to deviations from prescribed cooperation and ecological targets. Through equilibrium analysis, stability analysis, and phase-space characterization, we identify the principal operating regimes and transition mechanisms of the coupled system. The results show that the system admits a sustainable operating equilibrium in which cooperation, ecological resources, and institutional regulation coexist. We further establish the structural stability of a desirable operating state characterized by full cooperation, peak resource availability, and persistently minimal intervention. In addition to these stable regimes, the system exhibits bistability and limit cycles, indicating that alternative long-term outcomes may arise from different initial conditions and parameter configurations. These nonlinear mechanisms generate a persistent crackdown–rebound policy trap, in which intensified intervention temporarily improves cooperation and resource conditions but eventually produces recurrent oscillations and prevents system-wide consensus. The proposed framework provides a system-level explanation of governance failure and offers analytical insights for designing adaptive institutional regulation in coupled human–environment systems.
\end{abstract}

\begin{IEEEkeywords}
	Multi-agent system, evolutionary game theory, environmental feedback, equilibrium analysis, oscillations.
\end{IEEEkeywords}

%
\IEEEpeerreviewmaketitle

\section{Introduction}

\IEEEPARstart{T}{he} governance of common-pool resources (CPRs) presents a fundamental systems challenge arising from the classic tragedy of the commons \cite{hardin1998extensions}. At the center of this challenge is the interaction between two competing behavioral modes, namely, cooperation and free riding \cite{guo2023third,hu2021adaptive,wang2023evolutionary,zhang2024limitation,xu2025reinforcement,leung2026learning,ren2023reputation}. Cooperation requires individuals to voluntarily limit resource extraction or contribute to environmental maintenance, thereby incurring private costs for collective benefits \cite{paiva2018engineering}. In contrast, free riders seek to maximize their individual gains while avoiding the costs of resource conservation and environmental upkeep, thereby benefiting from the contributions of cooperators \cite{santos2023dynamics}. The resulting asymmetry between privately appropriated benefits and collectively shared degradation costs creates persistent incentives for free riding and undermines cooperative behavior \cite{santos2019evolution}. In the absence of effective institutional regulation, the interaction between behavioral incentives and ecological degradation can generate a reinforcing deterioration of the coupled social--ecological system, ultimately leading to the loss of cooperation and the depletion of the underlying resource base \cite{wang2022modelling}.

Evolutionary game theory provides a foundational mathematical framework for analyzing strategic interactions among individuals \cite{feng2023evolutionary,liu2024heterogeneously,lim2026truster,leung2026learning,han2022emergent,jia2025social,jia2025asymmetric,shi2025unified}. In recent years, to capture the reciprocal dynamics between human behaviors and natural resources, this field has advanced into the co-evolutionary game framework \cite{wang2023emergence,liu2023coevolutionary,stella2022lower,kawano2018evolutionary,hilbe2018evolution,liang2025coevolution}. In a seminal study, Weitz et al. integrated game-environment feedback into replicator dynamics \cite{weitz2016oscillating}. They demonstrated that this coupling generates an oscillating tragedy of the commons, where resource abundance and cooperative behavior cycle perpetually. Building on this paradigm, Tilman et al. generalized the incorporation of environmental feedbacks into evolutionary games, concluding that changing ecological states continuously alter the relative payoffs of strategies, thereby fundamentally reshaping evolutionary trajectories \cite{tilman2020evolutionary}. Driven by these foundational discoveries, the co-evolutionary framework has attracted increasing attention from researchers across various disciplines \cite{chen2025evolutionary,wang2026coevolutionary}.

Despite the explanatory power of these foundational models, a critical dimension of modern governance remains systematically overlooked \cite{pan2024towards,sun2023state,chen2022dim,ostrom2008challenge,chen2018punishment,liang2025coevolution}. Existing coevolutionary frameworks predominantly treat regulatory institutions as exogenous forces or perfectly rational central planners capable of instantaneous optimal interventions \cite{han2016emergence}. In reality, governing complex socioecological systems is profoundly constrained by bounded rationality and inherent decision making \cite{zhu2026more,eich2025bounded}. Rather than pursuing an absolute global optimum, regulatory bodies typically operate under the satisficing principle. They tend to scale back enforcement once a satisfactory target is met to minimize administrative friction \cite{hua2023facilitating,hua2024coevolutionary}. By neglecting these institutional realities, current models lack the mathematical capacity to fully explore why real world environmental policies frequently fail to sustain long term cooperation. Therefore, there is a pressing need to expand the traditional framework into a tripartite coevolutionary system that explicitly captures the timescale separation among rapid behavioral adaptation, slow ecological shifts, and lagged institutional responses.

In this paper, we develop a tripartite coevolutionary framework integrating social behavior, ecological resources, and institutional dynamics to understand how the inherent cost of cooperation drives system evolution. We design a dynamic feedback mechanism where regulatory intensity automatically increases when the system fails to meet explicit social and environmental targets. Through theoretic analysis, our study reveals three key findings.  

\begin{itemize}
\item First, the system exhibits a practical internal equilibrium where partial cooperation, sustainable resources, and active regulation naturally coexist.
\item Second, we demonstrate that an optimal state of full cooperation, peak resources, and no intervention can be structurally stable.
\item Finally, we show that the system displays complex dynamics, including bistability and structurally stable limit cycles, which fundamentally explain why environmental policies often fall into persistent crackdown and rebound oscillatory traps.

\end{itemize}
\section{Model and Methods}
\label{sec:model}

\subsection{Game-theoretic model}

Here we construst a tripartite co-evolutionary game model that couples social behavior, ecological resource state, and institutional regulatory intensity.
We consider an infinite population of individuals choosing between cooperation ($C$) and defection ($D$). Cooperators opt for sustainable resource extraction, while defectors choose overexploitation of resources. Let $x \in [0,1]$ denote the frequency of cooperators, and $1-x$ the frequency of defectors in the population. We assume that the payoffs of individuals are not static but are coupled to both the resource state and the institutional intervention. Based on prior studies \cite{weitz2016oscillating}, we set that abundant resources incentivize individuals to free-ride. In contrast, severely degraded resources raise the relative incentive to cooperate. Besides, we introduce institutional intervention to curb free-riding behavior. Let $g$ denote the intensity of institutional intervention or the probability of supervision. Individuals adopting the defection strategy will be penalized by institutions. Let \(\beta_1\) represent the penalty intensity coefficient of institutions. Accordingly, the expected institutional penalty cost borne by defectors equals \(\beta_1 g\).

Let $n \in [0, 1]$ represent the state of the environmental resource where $n=1$ denotes abundant carrying capacity and $n=0$ denotes severe depletion. We formulate the state-dependent effective payoff matrix $A(n, g)$ as:
$$\begin{aligned}
A(n, g)= (1-n)\begin{pmatrix} T & P \\ R-\beta_1 g & S-\beta_1 g \end{pmatrix} \\
+ n\begin{pmatrix} R & S \\ T-\beta_1 g & P-\beta_1 g \end{pmatrix},\end{aligned}$$
where $R$ denotes the mutual benefit obtained when both individuals choose to cooperate, $S$ is the diminished payoff an individual receives when cooperating while their counterpart defects, $T$ represents the highest possible payoff achieved by free-riding while the other player cooperates, and $P$ indicates the suboptimal baseline payoff when both players defect. These parameters satisfy the setup of the standard Prisoner's Dilemma, $T > R > P > S$. We further introduce bottom-up peer punishment, wherein cooperators within the group spontaneously impose peer sanctions on defectors, with a penalty cost of $kx$. Here, $k$ denotes the per-unit intensity cost of peer punishment. Then the expected payoffs for cooperators ($\pi_C$) and defectors ($\pi_D$) are derived as:
$$\begin{aligned}
\pi_C &= x[(1-n)T + nR] + (1-x)[(1-n)P + nS],\\
\pi_D &= x[(1-n)(R-\beta g) + n(T-\beta g)] + (1-x)[(1-n)(S-\beta g) \\
&+ n(P-\beta g)]-kx.
\end{aligned}$$

The frequency of cooperator $x$ evolves according to the replicator equation \cite{schuster1983replicator,cressman2014replicator}, 
$$\begin{aligned}
	\dot{x} &= x(1-x)(\pi_C - \pi_D) \\
	&= x(1-x)\{(1 - 2n) \left[ x\Delta_{TR} + (1-x)\Delta_{PS} \right] + \beta_1 g + \kappa x\},
\end{aligned}$$
where $\Delta_{TR}=T-R$ denotes the relative payoff advantage of defection over cooperation when facing a cooperative opponent. $\Delta_{PS}=P-S$ denotes the relative payoff advantage of defection over cooperation when facing a defective opponent \cite{ito2024complete}.

Individual behavior directly affects environmental state. Specially, cooperation improves environmental quality, while defection causes degradation. Given the positive role of policy incentives, we define the environmental state equation as:
$$\dot{n} = n(1-n) \left[ \theta x-(1-x) + \eta g \right],$$
where $\eta g$ captures the direct positive impact of government incentives on resource restoration, \(\theta x\) denotes the positive effect generated by social cooperation, while \(-(1-x)\) represents the negative effect stemming from defection.

We introduce dual thresholds for the intensity of government incentives. Specifically, the incentive intensity rises if the environmental resource stock falls below the target value \(n_{\text{target}}\), and it also increases when the proportion of cooperators drops below the target level \(x_{\text{target}}\). Its dynamic equation can be formulated as follows:
$$\dot{g} = g(1-g) \left[ \alpha(n_{target} - n) + \beta_2(x_{target} - x) - \delta g \right],$$
where \(\alpha\) and \(\beta_2\) respectively measure the government’s sensitivity to ecological conditions and the level of social cooperation. \(-\delta g\) characterizes the incentive cost. For clarity, we summarize the model parameters and their meanings in Table \ref{tab:parameters}.
\begin{table}[ht]
	\centering
	\footnotesize
	\caption{Model parameters and their definitions}
	\label{tab:parameters}
	\begin{tabular}{cl}
		\hline
		\textbf{Parameter} & \textbf{Definition} \\
		\hline
		$\epsilon_1$ & Time-scale for environmental dynamics \\
		$\epsilon_2$ & Time-scale for institutional dynamics \\
		$\kappa$ & Peer pressure intensity \\
		$\beta_1$ & Effect of intervention on cooperation \\
		$\alpha$ & Sensitivity of intervention to ecological deficits \\
		$\beta_2$ & Sensitivity of intervention to cooperation deficits \\
		$\delta$ & Maintenance cost of intervention \\
		$n_{target}$ & Ecological target threshold \\
		$x_{target}$ & Social target threshold  \\
		$\theta$ & Contribution of cooperation to ecological recovery \\
		$\eta$ & Contribution of intervention to the environment \\
		\hline
	\end{tabular}
\end{table}

\subsection{The coevolutionary dynamical system}
\label{subsec:dynamical_system}

The social-ecological system described above operates on heterogeneous time scales, where individual strategy shifts often occur more rapidly than environmental restoration or policy adjustments. To capture this, we introduce the reaction rates $\epsilon_1$ and $\epsilon_2$ into the ecological state equations and government intervention. The complete social-ecological dynamical system is defined as follows:

\begin{equation*}
	\label{eq:dynamical_system}
	\begin{cases} 
		\dot{x} =  x(1-x) \{(1 - 2n) \left[ x\Delta_{TR} + (1-x)\Delta_{PS} \right] + \beta_1 g + \kappa x\}, \\[8pt]
		\dot{g} = \epsilon_2 g(1-g)  \left[ \alpha(n_{target} - n) + \beta_2(x_{target} - x) - \delta g \right], \\[8pt]
		\dot{n} =\epsilon_1 n(1-n) \left[ \theta x-(1-x) + \eta g \right],
	\end{cases}
\end{equation*}
where $0<\epsilon_1< \epsilon_2 \ll 1$ define the evolutionary tempo of the system.  In the next section, we provide a detailed analysis of the evolutionary dynamics of this coupled system.
\section{Results}
In this section we provide a detailed theoretical analysis of the above tripartite coevolutionary game model and provide numerical experiments for verification. The argument will be elaborated through some progressive theorems: firstly, analyze the monostable dynamics of the system, including corner equilibrium points, boundary equilibrium points, surface equilibrium points, and internal equilibrium points. Subsequently, analyze the bistable dynamical results of the system. Finally, the condition for the occurrence of periodic oscillations is provided. 

\begin{theorem}
	For a tripartite coevolutionary dynamical system defined in the three-dimensional state space $\Omega=\{(x, g, n) \mid 0 \le x, g, n \le 1 \}$, there are always eight corner equilibria, which are:
	\begin{equation*}
		\begin{aligned}
			E_1(0,0,0), \quad & E_2(1,0,0), \quad & E_3(0,1,0), \quad & E_4(0,0,1) \\
			E_5(1,1,0), \quad & E_6(1,0,1), \quad & E_7(0,1,1), \quad & E_8(1,1,1).
		\end{aligned}
	\end{equation*}
	Among all eight corner equilibrium points, $E_6(1,0,1)$ is stable when $\Delta_{TR}<\kappa$ and $E_7(0,1,1)$ is stable when $\beta_1<\Delta_{PS}$, $\alpha(n_{target}-1)+\beta_2 x_{target}-\delta>0$, and $1-\eta<0$.
\end{theorem}

\begin{proof}
	Let the replicator equation system satisfy $\frac{dx}{dt}=0, \frac{dg}{dt}=0, \frac{dn}{dt}=0 $. Due to the fact that the evolution rates of state variables $x, g$, and $n$ contain factors $x (1-x), g (1-g)$, and $n (1-n) $, respectively, when the variable values are $0$ or $1$, the evolution rates of each dimension of the system are always zero. Therefore, the above eight corner must be fixed points of the system.
	
	To analyze local asymptotic stability, we can calculate the Jacobian matrix $J$ of the system at any point $(x, g, n)$. Substituting eight corner equilibrium points into the matrix $J$ makes it easy to determine that the equilibrium points of the other six vertices are unstable, as their Jacobian matrix has at least one positive eigenvalue. We focus on examining two key boundary points, $E_6 (1,0,1)$ and $E_7 (0,1,1)$, that can become evolutionarily stable. For ideal autonomous state $E_6 (1,0,1) $: By substituting its coordinates into the Jacobian matrix, three eigenvalues can be obtained, which are
			$\lambda_{6,1} = \Delta_{TR} - \kappa,$
			$\lambda_{6,2} =\epsilon_2 \left[ \alpha(n_{target} - 1) + \beta_2(x_{target} - 1) \right],$
			$\lambda_{6,3} =-\epsilon_1 \theta.$
	Therefore, the condition for $E_6$ to become a locally asymptotically stable point is $\Delta_{TR} < \kappa$.
	For $E_7 (0,1,1)$, the eigenvalues are
		$	\lambda_{7,1}=\beta_1 - \Delta_{PS},$
		$	\lambda_{7,2}=\epsilon_2 \left[ \delta - \alpha(n_{target} - 1) - \beta_2 x_{target} \right],$
			$\lambda_{7,3}=\epsilon_1(1 - \eta).$
	The condition for $E_7$ to become another ESS in the system is $\lambda_{7,1}<0 $, $\lambda_{7,2}<0 $, $\lambda_{7,3}<0 $.
\end{proof}

\begin{figure}[h]
	\centering 
	\includegraphics[width=0.5\textwidth]{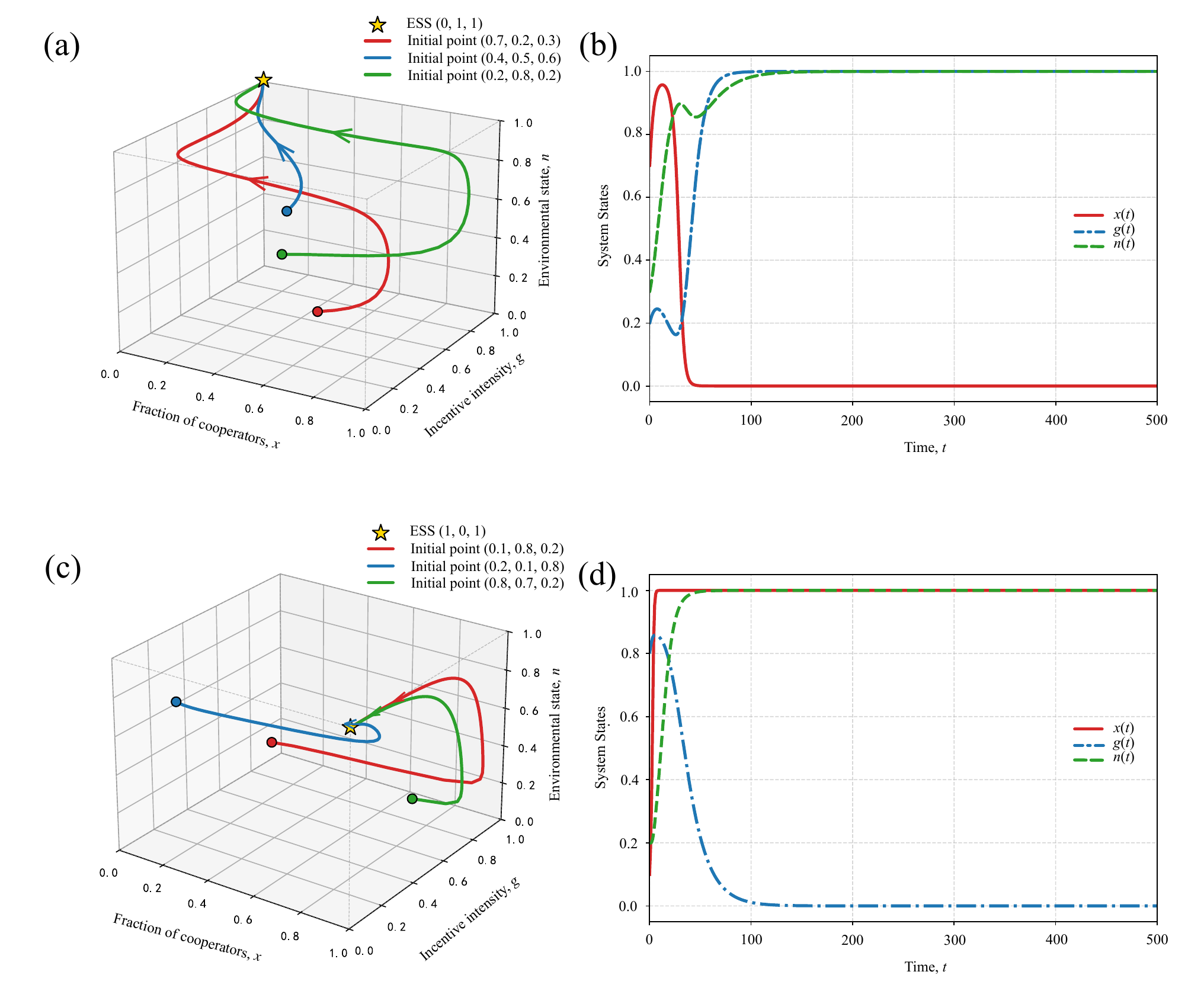}  
	\caption{\textbf{Global mono-stability of the corner equilibrium points.} 3D phase portrait and corresponding frequency-time plot illustrating the global stability of the autonomous ideal state $(1, 0, 1)$ for panels (a) and (b), and $(0, 1, 1)$ for panels (c) and (d). The parameter values are  $\kappa = 0.1$, $\Delta_{PS} = 0.6$, $\beta_1 = 0.2$, $\beta_1 = 0.8$, $\eta = 1.5$ in panels (a) and (b); for (c-d) $\kappa = 0.8$,  $\Delta_{PS} = 0.3$, $\beta_1 = 0.5$, $\beta_1 = 0.5$, $\eta = 0.8$. Other baseline parameters remain consistent: $\Delta_{TR} = 0.5$, $\epsilon_1=0.1$, $\epsilon_2=0.3$, $\alpha=0.5$, $x_{target}=0.8$, $n_{target}=0.8$, $\theta=1$, $\delta=0.2$.
	}
	\label{figure1}  
\end{figure} 
In Figure \ref{figure1}, we provide numerical examples to validate the theoretical analysis results mentioned above. In the first column of Figure \ref{figure1}, we present a 3D phase diagram, while the second column shows the evolution of the system state over time. When the model parameters meet $\beta_1<\Delta_{PS}$, $\alpha(n_{target}-1)+\beta_2 x_{target}-\delta>0$, and $1-\eta<0$, we find that curves starting from different initial points eventually converge to $E_7 (0,1,1)$, which means that even if the environment is abundant and the government adopts the strongest control measures, all cooperators disappear (Figure \ref{figure1}(a) and (b)). In addition, when the model parameters satisfy $\Delta_{TR}<\kappa $, we find that the trajectories of the phase space eventually converge to $E_6 (1,0,1)$, which means that even if the external incentive intensity is 0, peer punishment can still encourage all individuals to choose cooperation, and the environment is in a state of abundance (Figure \ref{figure1}(c) and (d)).

\begin{figure}[h]
	\centering 
	\includegraphics[width=0.5\textwidth]{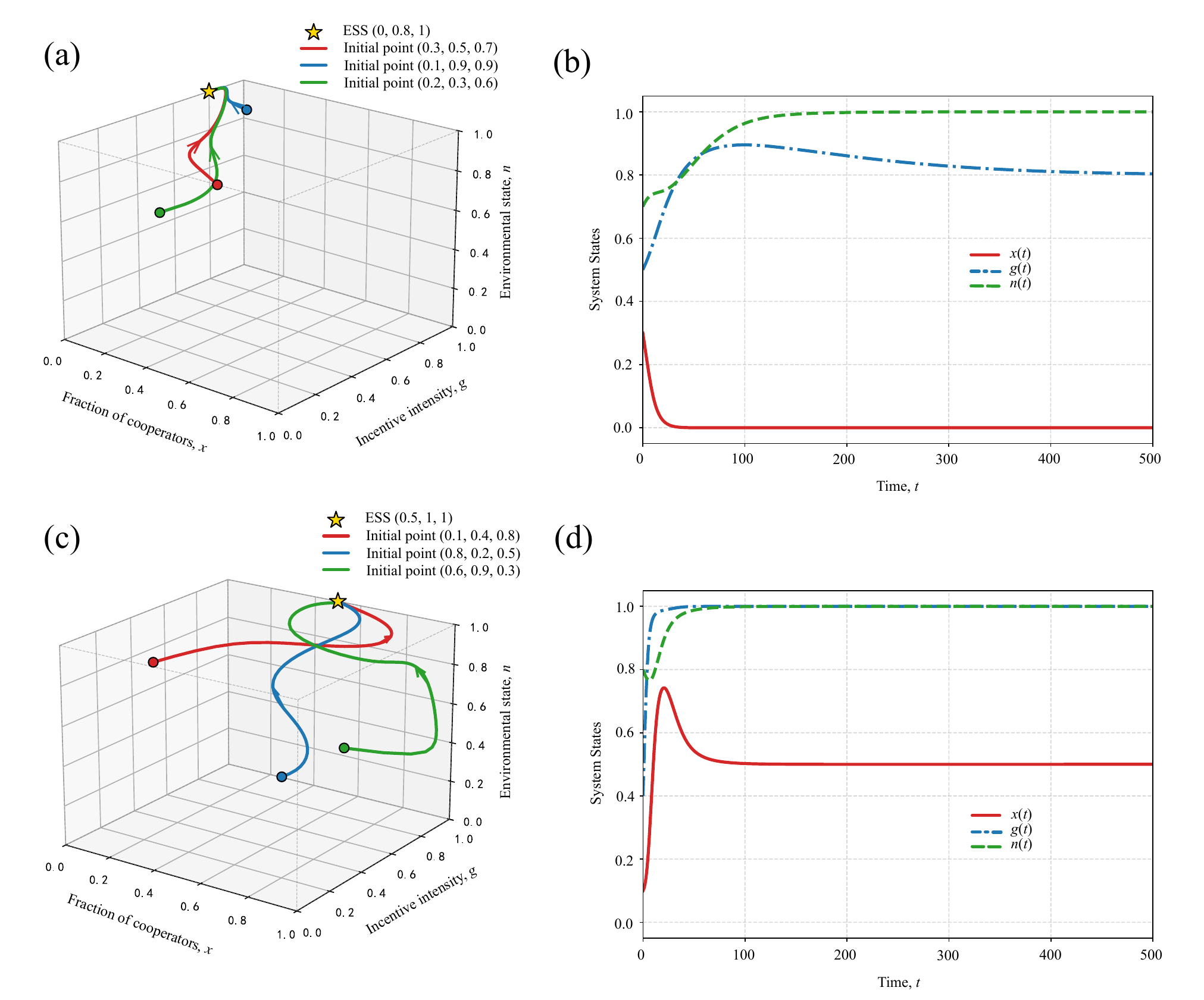}  
	\caption{\textbf{Global mono-stability and temporal dynamics of the boundary equilibrium points.} The 3D phase portrait and corresponding frequency-time plot for the boundary equilibrium $(0, 0.8, 1)$ in panels (a) and (b), $(0.5, 1, 1)$ for panels (c) and (d). Parameters are $\Delta_{TR} = 0.5$, $\Delta_{PS} = 0.6$, $\beta_1 = 0.2$, $\beta_2 = 0.4$, $\alpha=0.5$, $\eta = 1.5$, $x_{target}=0.65$, $\delta=0.2$ in panels (a) and (b); for (c-d) $\Delta_{TR} = 0.9$, $\Delta_{PS} = 0.2$, $\beta_1 = 0.5$, $\beta_2 = 2$, $\alpha=0.2$, $\eta = 0.5$, $x_{target}=0.9$, $\delta=0.1$. Other parameters are $\kappa = 0.1$, $\epsilon_1=0.1$, $\epsilon_2=0.3$, $n_{target}=0.8$, $\theta=1$.
	}
	\label{figure2}  
\end{figure} 
In addition to the eight corner equilibrium points mentioned above, there are also some boundary equilibrium points in the coupled system mentioned above, such as $(0, \frac{\alpha n_{target}+\beta_2 x_{target}}{\delta}, 0)$, $(1, \frac{\alpha n_{target}+\beta_2 (x_{target}-1)}{\delta}, 0)$, $(0, \frac{\alpha (n_{target}-1)+\beta_2 x_{target}}{\delta}, 1)$, $(\frac{\Delta_{PS}}{\Delta_{PS}-\Delta_{TR}+\kappa}, 0,1)$, $(\frac{\Delta_{PS}-\beta_1}{\Delta_{PS}-\Delta_{TR}+\kappa}, 1, 1)$. In the following theorem, we provide the conditions for the stability of boundary equilibrium points $(0, \frac{\alpha (n_{target}-1)+\beta_2 x_{target}}{\delta}, 1)$ and $(\frac{\Delta_{PS}-\beta_1}{\Delta_{PS}-\Delta_{TR}+\kappa}, 1, 1)$.
\begin{theorem}\label{theo3.2}
	If and only if $0<\alpha (n_{target} -1)+\beta_2 x_{target}<\delta$ is satisfied, the equilibrium point $(0, \frac{\alpha (n_{target}-1)+\beta_2 x_{target}}{\delta}, 1)$ exists.
	The condition for this point to be locally asymptotically stable (ESS) is $	\beta_1 g ^ *<\Delta_{PS}$ and $\eta g^*>1$.
	
	If and only if $0<\frac{\Delta_{PS} - \beta_1} {\Delta_{PS} - \Delta_{TR}+\ kappa}<1$ is satisfied, the equilibrium point $(\frac{\Delta_{PS}-\beta_1}{\Delta_{PS}-\Delta_{TR}+\kappa}, 1, 1)$ exists.
	The condition for this point to be locally asymptotically stable (ESS) is $\Delta_{PS} - \Delta_{TR} + \kappa < 0,  \delta < \alpha(n_{target} - 1) + \beta_2(x_{target} - \frac{\Delta_{PS}-\beta_1}{\Delta_{PS}-\Delta_{TR}+\kappa})$ and  $(\theta+1)\frac{\Delta_{PS}-\beta_1}{\Delta_{PS}-\Delta_{TR}+\kappa}+\eta>1$.
\end{theorem}
\begin{proof}
	Let $x=0, n=1, \dot{g}=0 $, and solve for the expression $g^*= \frac{\alpha (n_{target}-1)+\beta_2 x_{target}}{\delta} $. Substituting $(0, g^*, 1)$ into $J$ yields the eigenvalues $\lambda_{8,1} = \beta_1 g^* - \Delta_{PS}$, $\lambda_{8,2} = - \epsilon_2 g^*(1-g^*)$, $\lambda_{8,3} = \epsilon_1(1 - \eta g^*)$.
	Due to $g^* \in (0,1) $, $\lambda_{8,2}$ is always less than 0. By setting $\lambda_{8,1}<0$ and $\lambda_{8,3}<0$, we obtain $1/\eta<g ^ *<\Delta_{PS}/\beta_1 $. By combining the natural premise of $g^* \in (0,1) $and substituting it into the algebraic formula of $g^*$, conclusion can be obtained.
	
	Let $g=1, n=1, \dot {x}=0 $, and solve for the expression $x^*=\frac{\Delta_{PS}-\beta_1}{\Delta_{PS}-\Delta_{TR}+\kappa} $. Substituting the Jacobian matrix yields
	$\lambda_{9,1} =  x^*(1-x^*) (\Delta_{PS} - \Delta_{TR} + \kappa)$,
	$\lambda_{9,2} = \epsilon_2\left[ \delta - \alpha(n_{target} - 1) - \beta_2(x_{target} - x^*) \right]$,
	$\lambda_{9,3} = \epsilon_1[1 - \eta - (\theta + 1)x^*]$.
	To make $\lambda_{9,1}<0 $, there must be $(\Delta_ {PS} - \Delta_ {TR}+\kappa)<0 $. At this point, in order to ensure the existence of $0<x^*<1 $, its molecule must have $\Delta_{PS} - \beta_1<0$ and $\Delta_{PS} - \beta_1 > \Delta_{PS} - \Delta_{TR} + \kappa$. Let $\lambda_{9,2}, \lambda_{9,3}<0$ to obtain the conclusion. Theorem \ref{theo3.2} has been proven.
\end{proof}

Next, we provide numerical examples to validate the theoretical analysis mentioned above. Figure \ref{figure2} presents the phase diagram of a three-dimensional coupled game system and the evolution of the system states over time when the boundary equilibrium points $(0, \frac{\alpha (n_{target}-1)+\beta_2 x_{target}}{\delta}, 1) $ and $(\frac{\Delta_{PS}-\beta_1}{\Delta_{PS}-\Delta_{TR}+\kappa}, 1, 1)$ are stable. When the model parameters satisfy $0<\alpha (n_{target} -1)+\beta_2 x_{target}<\delta$, we find that there exists a boundary equilibrium point $(0, 0.8, 1)$, and all internal trajectories in the phase diagram converge to this equilibrium point when parameters satisfy $\beta_1 g ^ *<\Delta_{PS}$ and $\eta g^*>1$ (Figure \ref{figure2}(a) and (b)). We turn to another boundary equilibrium point, and when the model parameters satisfy $0<\frac{\Delta_{PS} - \beta_1} {\Delta_{PS} - \Delta_{TR}+\ kappa}<1$, we find that the boundary equilibrium point $(0.5, 1, 1)$ exists, and when the parameters satisfy  $\Delta_{PS} - \Delta_{TR} + \kappa < 0,  \delta < \alpha(n_{target} - 1) + \beta_2(x_{target} - x^*)$ and  $(\theta+1) x ^ *+\ eta>1$, all trajectories converge to this point (Figure \ref{figure2}(c) and (d)). The former means that when the government has a high-intensity regulatory system, it can ensure abundant resources but all individuals choose free riding behavior. The latter means that the government fully supervises, ensures a constant level of cooperation, and resources are always abundant.

\begin{figure}[t!]
	\centering 
	\includegraphics[width=0.5\textwidth]{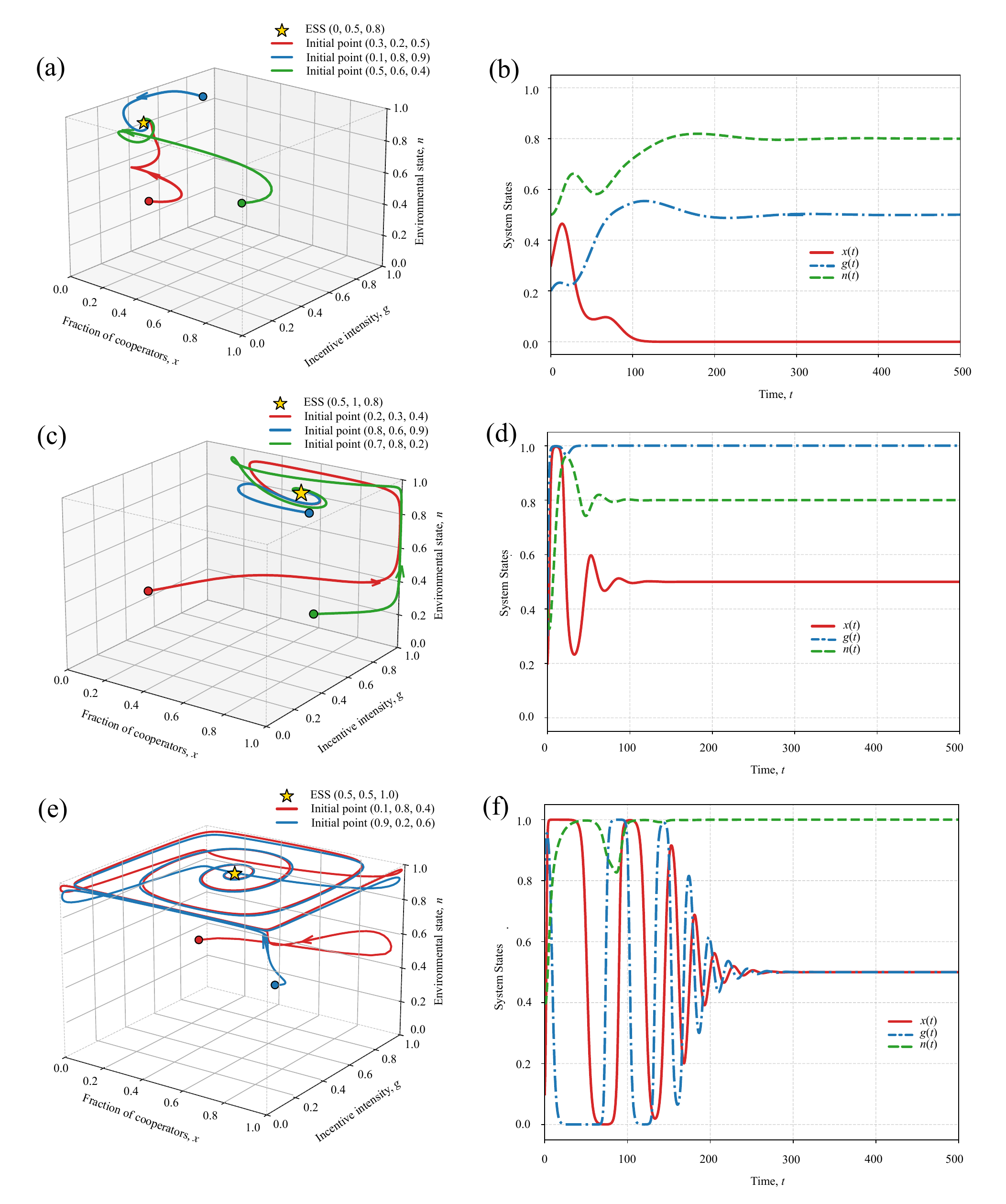}  
	\caption{\textbf{Global mono-stability of the face equilibrium points.} The 3D phase portraits and corresponding frequency-time plots illustrate that the system asymptotically converges to the respective face equilibria $(0, 0.5, 0.8)$ for (a-b), $(0.5, 1, 0.8)$ for (c-d), and $(0.5, 0.5, 1.0)$ for (e-f). Parameters are $\Delta_{TR} = 0.6$, $\Delta_{PS} = 0.5$, $\beta_1 = 0.2$, $\beta_2 = 0.5$, $\alpha=0.5$, $\eta = 2$, $x_{target}=0.6$, $\delta=0.4$,  $\kappa = 0.1$, $\epsilon_1=0.1$, $\epsilon_2=0.3$, $n_{target}=0.6$, $\theta=1$ in panels (a) and (b); for (c-d) $\Delta_{TR} = 0.8$, $\Delta_{PS} = 0.2$, $\beta_1 = 0.25$, $\beta_2 = 1$, $\alpha=2$, $\eta = 0.4$, $x_{target}=0.8$, $\delta=0.2$,  $\kappa = 0.1$, $\epsilon_1=0.2$, $\epsilon_2=0.5$, $n_{target}=0.9$, $\theta=0.2$; for (e-f)  $\Delta_{TR} = 0.5$, $\Delta_{PS} = 0.2$, $\beta_1 = 0.4$, $\beta_2 = 1.5$, $\alpha=1$, $\eta = 0.8$, $x_{target}=0.7$, $\delta=0.4$,  $\kappa = 0.6$, $\epsilon_1=0.1$, $\epsilon_2=0.5$, $n_{target}=0.9$, $\theta=0.6$.
	}
	\label{fig3}  
\end{figure} 
On the boundary plane of the three-dimensional state space $\Omega$ of the tripartite co-evolutionary dynamical system, there are four surface equilibrium points $E(0, \frac{1}{\eta}, n_{target}+\frac{\beta_2 x_{target}-\delta/ \eta}{\alpha}) $, $E(\frac{1-\eta}{1+\theta}, 1, 1/2+\frac{\beta_1 +\kappa \frac{1-\eta}{1+\theta}}{2[\frac{1-\eta}{1+\theta} \Delta_{TR}+(1-\frac{1-\eta}{1+\theta})\Delta_{PS}]})$, $E(x_1^*, g_1^*, 1)$, and $E(\frac{1}{1+\theta}, 0, 1/2+\frac{\kappa \frac{1}{1+\theta}}{2[\frac{1}{1+\theta} \Delta_{TR}+(1-\frac{1}{1+\theta})\Delta_{PS}]})$, where $x_1^*$ and $g_1^*$ are given by the solution of the following system of linear equations, and must satisfy $x_1^*, g_1^* \in (0,1)$:
\begin{equation}
	\begin{cases}
		(\kappa - \Delta_{TR} + \Delta_{PS})x_1^* + \beta_1 g_1^* = \Delta_{PS} \\
		\beta_2 x_1^* + \delta g_1^* = \alpha(n_{target} - 1) + \beta_2 x_{target}.
	\end{cases}
\end{equation}
We can obtain
$x_1^*=\frac{\Delta_{PS} \delta - \beta_1 [\alpha(n_{target} - 1) + \beta_2 x_{target}]}{\delta(\kappa - \Delta_{TR} + \Delta_{PS}) - \beta_1 \beta_2}$ and $g_1^*=\frac{(\kappa - \Delta_{TR} + \Delta_{PS})[\alpha(n_{target} - 1) + \beta_2 x_{target}] - \Delta_{PS} \beta_2}{\delta(\kappa - \Delta_{TR} + \Delta_{PS}) - \beta_1 \beta_2}$.
The existence of these equilibrium points and the conditions for becoming evolutionarily stable are present in the following Theorem.

\begin{theorem}\label{theo3.3}
	The equilibrium point $E(0, \frac{1}{\eta}, n_{target}+\frac{\beta_2 x_{target}-\delta/ \eta}{\alpha})$ exists only when the system parameters satisfy $\eta>1$ and $n_{target}+\frac{\beta_2 x_{target}-\delta/ \eta}{\alpha} \in (0,1)$ is obtained. The sufficient and necessary condition for this point to be locally asymptotically stable is $(1 - 2n_{target}+2\frac{\beta_2 x_{target}-\delta/ \eta}{\alpha})\Delta_{PS} + \frac{\beta_1}{\eta} < 0$.\\
	The equilibrium point $E(\frac{1-\eta}{1+\theta}, 1, 1/2+\frac{\beta_1 +\kappa \frac{1-\eta}{1+\theta}}{2[\frac{1-\eta}{1+\theta} \Delta_{TR}+(1-\frac{1-\eta}{1+\theta})\Delta_{PS}]})$ exists when $0<\frac{1-\eta}{1+\theta}<1$ and $0<1/2+\frac{\beta_1 +\kappa \frac{1-\eta}{1+\theta}}{2[\frac{1-\eta}{1+\theta} \Delta_{TR}+(1-\frac{1-\eta}{1+\theta})\Delta_{PS}]})<1$. The stability of this equilibrium point requires the following conditions to be met
	\begin{equation}\label{eq7}
		\begin{cases}
			\alpha(n_{target} -1/2+\frac{\beta_1 +\kappa \frac{1-\eta}{1+\theta}}{2[\frac{1-\eta}{1+\theta} \Delta_{TR}+(1-\frac{1-\eta}{1+\theta})\Delta_{PS}]}) \\+ \beta_2(x_{target} - \frac{1-\eta}{1+\theta}) > \delta \\
			\frac{1-\eta}{1+\theta}\Delta_{TR} + (1-\frac{1-\eta}{1+\theta})\Delta_{PS} > 0 \\
			-\frac{\beta_1 +\kappa \frac{1-\eta}{1+\theta}}{\frac{1-\eta}{1+\theta} \Delta_{TR}+(1-\frac{1-\eta}{1+\theta})\Delta_{PS}}(\Delta_{TR} - \Delta_{PS}) + \kappa < 0.
		\end{cases}
	\end{equation}
	The equilibrium point $E(x_1^*, g_1^*, 1)$ exists when $x_1^*, g_1^* \in (0,1)$. The stability of this equilibrium point requires the following conditions to be met
	\begin{equation}\label{eq8}
		\begin{cases}
			\theta x_1^* - (1-x_1^*) + \eta g_1^* > 0 \\
			\delta(\kappa - \Delta_{TR} + \Delta_{PS}) - \beta_1 \beta_2 < 0 \\
			\frac{1}{\epsilon_1} x_1^*(1-x_1^*)(\Delta_{PS} - \Delta_{TR} + \kappa) - \frac{\delta}{\epsilon_2} g_1^*(1-g_1^*) < 0.
		\end{cases}
	\end{equation}
	$E(\frac{1}{1+\theta}, 0, \frac{1}{2}+\frac{\kappa \frac{1}{1+\theta}}{2[\frac{1}{1+\theta} \Delta_{TR}+(1-\frac{1}{1+\theta})\Delta_{PS}]})$ is unstable.
\end{theorem}
\begin{proof}
	For $E(0, \frac{1}{\eta}, n_{target}+\frac{\beta_2 x_{target}-\delta/ \eta}{\alpha})$,
	let $x=0, \dot{g}=0, \dot{n}=0$. we can obtain $g_2^*=1/\eta $. Substituting $\dot{g}=0 $ yields the expression for $n_2^*$.
	Find the partial derivative at $x=0$, where the single dimensional eigenvalue $\lambda_1= (1-2n_2^*) \Delta_{PS}+\beta_1 g_2^*$. For the submatrix of the $(g, n)$ space:
	$\frac{\partial \dot{g}}{\partial g} = -\epsilon_2\delta g_2^*(1-g_2^*) < 0$, $\frac{\partial \dot{n}}{\partial n} = 0$, Therefore, $\text{Tr} (N_{sub})<0 $ always holds.
	Also, since the product of interleaved partial derivatives is negative , namely, $\frac{\partial \dot{g}} {\partial n}<0$ and $\frac{\partial \dot{n}} {\partial g}>0$, $\text{Det} (J_{sub})>0$ always holds. Therefore, the stability of the equilibrium point on this surface depends entirely on $\lambda_1<0 $, that is, $(1 - 2n_{target}+2\frac{\beta_2 x_{target}-\delta/ \eta}{\alpha})\Delta_{PS} + \frac{\beta_1}{\eta} < 0$.
	
	For $E(\frac{1-\eta}{1+\theta}, 1, 1/2+\frac{\beta_1 +\kappa \frac{1-\eta}{1+\theta}}{2[\frac{1-\eta}{1+\theta} \Delta_{TR}+(1-\frac{1-\eta}{1+\theta})\Delta_{PS}]})$,
	let $g=1, \dot {x}=0, \dot {n}=0$, we can obtain $x_3^*=\frac{1-\eta}{1+\theta}$. Substituting the core driver term $\dot{x}=0 $ yields $n_3^*= 1/2+\frac{\beta_1 +\kappa \frac{1-\eta}{1+\theta}}{2[\frac{1-\eta}{1+\theta} \Delta_{TR}+(1-\frac{1-\eta}{1+\theta})\Delta_{PS}]}$.
	Boundary eigenvalue $\lambda_2 = -\epsilon_2[\alpha(n_{target} - n_3^*) + \beta_2(x_{target} - x_3^*) - \delta]$. Let $\lambda_2<0$ to obtain the first equation of conditional expression (\ref{eq7}).
	For the $(x, n)$ submatrix, the requirement is $\text{Det}(J_{sub}) = - (\frac{\partial \dot{x}}{\partial n})(\frac{\partial \dot{n}}{\partial x}) > 0$. Since $\frac {\partial \dot{n}} {\partial x}>0 $, there must be $\frac{\partial \dot{x}} {\partial n}<0 $, simplifying $x_3^* \Delta_{TR}+(1-x_3^*) \Delta_{PS}>0 $. Let $\text{Tr} (N_{sub})=\frac{\partial \dot {x}} {\partial x}<0 $, and we will obtain the third equation of conditional expression (\ref{eq7}).
	
	For $E(x_1^*, g_1^*, 1)$, let $n=1, \dot{x}=0, \dot{g}=0$, we can obtain $x_1^*$ and $g_1^*$.
	The boundary eigenvalue $\lambda_3=-\epsilon_1[\theta x_1^* - (1-x_1^*)+\eta g_1^*]$, where $\lambda_3<0$, provides the first stability condition of equation (\ref{eq8}). For the submatrix of the internal $(x, g)$ space, its determinant is calculated as:
	\begin{equation*}
		\text{Det}(J_{sub}) = \epsilon_2 x_1^*(1-x_1^*)g_1^*(1-g_1^*) \left[ \beta_1 \beta_2 - \delta(\kappa - \Delta_{TR} + \Delta_{PS}) \right].
	\end{equation*}
	Due to $x_1^*, g_1^* \in (0,1) $, requiring $\text{Det} (J{sub})>0$ is equivalent to the denominator of the characteristic analytical solution being strictly less than 0. Combined with trace $\text{Tr} (J{sub})<0 $, it is integrated into the comprehensive stability judgment condition in equation (\ref{eq8}). \\
	For $E(\frac{1}{1+\theta}, 0, \frac{1}{2}+\frac{\kappa \frac{1}{1+\theta}}{2[\frac{1}{1+\theta} \Delta_{TR}+(1-\frac{1}{1+\theta})\Delta_{PS}]})$, in the $g=0$ subspace, we know $\frac{\partial \dot{n}}{\partial n} \equiv 0$, making the reduced Jacobian's trace completely determined by $\frac{\partial \dot{x}}{\partial x}$. Substituting the algebraic equilibrium condition ($\dot{x}=0$) simplifies this trace exactly to $\text{Tr}(J_{sub}) = \frac{\theta}{1+\theta}(\frac{\kappa \frac{1}{1+\theta}}{[\frac{1}{1+\theta} \Delta_{TR}+(1-\frac{1}{1+\theta})\Delta_{PS}]})\Delta_{PS}>0$, it mathematically guarantees $\text{Tr}(J_{sub}) > 0$. Therefore, it is unstable.
	Theorem \ref{theo3.3} has been proven.
\end{proof}	

Three representative numerical examples are presented in Figure \ref{fig3} to validate the theoretical analysis mentioned above. Figures \ref{fig3}(a) and (b) present a scenario where $E(0, \frac{1}{\eta}, n_{target}+\frac{\beta_2 x_{target}-\delta/ \eta}{\alpha})$ exists and is stable when the model parameters satisfy the conditions of Theorem \ref{theo3.3}. We find that all trajectories in the phase space converge to this stable equilibrium point, which means that maintaining a constant incentive intensity can ensure sustainable resources, but everyone chooses free riding behavior. Figures \ref{fig3}(c) and (d) present a scene where $E(\frac{1-\eta}{1+\theta}, 1, 1/2+\frac{\beta_1 +\kappa \frac{1-\eta}{1+\theta}}{2[\frac{1-\eta}{1+\theta} \Delta_{TR}+(1-\frac{1-\eta}{1+\theta})\Delta_{PS}]})$ exists and is stable. This means that maximizing the incentive intensity can ensure a constant level of individual choice for cooperative behavior and sustainable resources. Figures \ref{fig3}(e) and (f) present a scene where $E(x_1^*, g_1^*, 1)$ exists and is stable. At this point, a constant incentive intensity ensures a stable level of cooperation and resources are always abundant.

\begin{figure}[h]
	\centering 
	\includegraphics[width=0.5\textwidth]{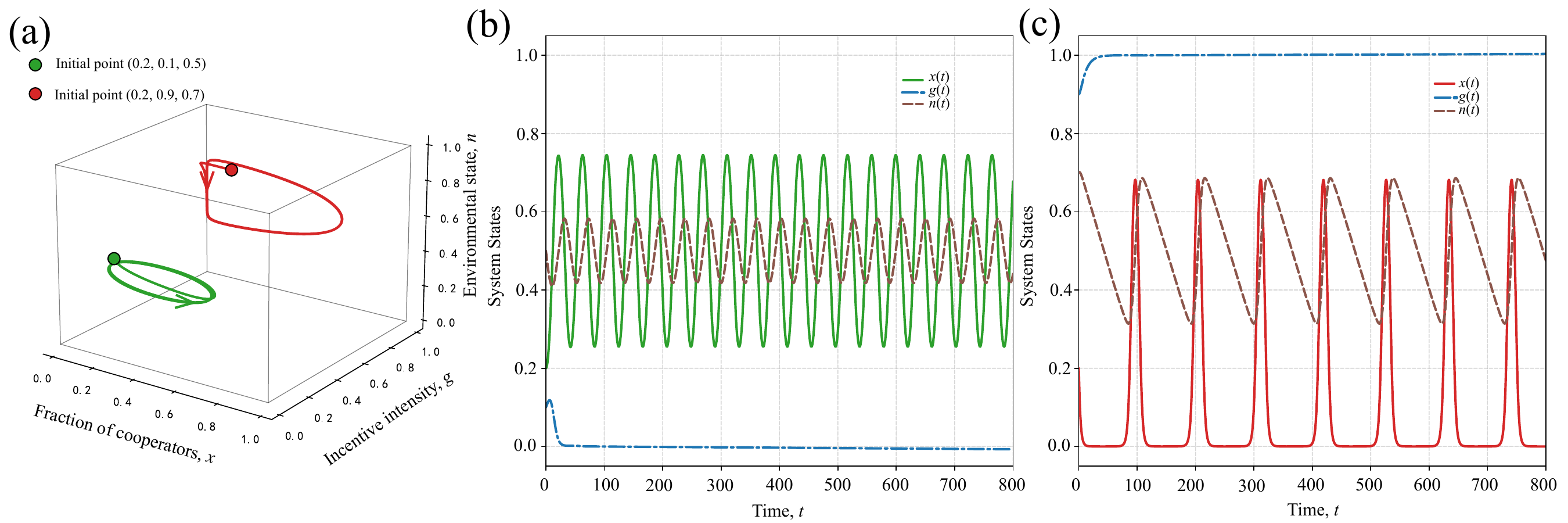}  
	\caption{\textbf{Coexistence of dual boundary limit cycles in the tripartite co-evolutionary system.} Panel (a) show 3D phase space illustrating two distinct periodic closed orbits. Trajectories converge to different boundary limit cycles based on initial conditions. Panel (b) show time-series of the green trajectory starting at (0.2, 0.1, 0.5). Panel (c) show time-series of the red trajectory starting at (0.2, 0.9, 0.7).  Parameters are $\Delta_{TR} = 1$, $\Delta_{PS} = 1$, $\beta_1 = 0$, $\beta_2 = 1$, $\alpha=0$, $\eta = 0.8$, $x_{target}=0.3$, $\delta=0.1$,  $\kappa = 0$, $\epsilon_1=0.1$, $\epsilon_2=0.5$, $n_{target}=0.5$, $\theta=1$.
	}
	\label{fig4}  
\end{figure} 
\begin{theorem}\label{theo3.4}
	On the boundary plane ($g \in \{0,1 \} $) of the coupling system, if $\kappa=0$ and $\beta_1=0$, the two-dimensional subsystem degenerates into a conservative system. At this point, the evolution trajectory of the system presents a closed periodic orbit. 
\end{theorem}

\begin{proof}	
	Under the condition of $g \in \{0,1 \} $ and $\beta_1=\kappa=0$, the evolutionary equation of the system can be expressed as:
	\begin{equation}
		\begin{cases} 
			\dot{x} &= x(1-x)(1-2n) \left[ \Delta_{PS} + x(\Delta_{TR} - \Delta_{PS}) \right]\\
			\dot{n} &= \epsilon_1 n(1-n) \left[ (\theta+1)x - 1 + \eta g \right].
		\end{cases}
	\end{equation}
	By separating variables, we have
	$$\frac{(\theta+1)x - 1 + \eta g}{x(1-x) \left[ \Delta_{PS} + (\Delta_{TR} - \Delta_{PS})x \right]} dx = \frac{1-2n}{\epsilon_1 n(1-n)} dn.$$
	
	Partial fractional expansion and integration can be performed on both ends of the above equation to construct a generalized motion integral $H(x,n)$
	\begin{equation*}
	\begin{aligned}
			H(x,n) &= -\frac{1}{\epsilon_1} \ln \left( n(1-n) \right) - \frac{1-\eta g}{\Delta_{PS}}\ln x - \frac{\theta + \eta g}{\Delta_{TR}}\ln(1-x) \\
			&+ \left(\frac{\theta + \eta g}{\Delta_{TR}} + \frac{1-\eta g}{\Delta_{PS}}\right)\ln\left[\Delta_{PS} + (\Delta_{TR} - \Delta_{PS})x\right].
				\end{aligned}
	\end{equation*}

	The full derivative of the energy function along the system evolution trajectory is always zero, that is, $\dot {H} \equiv 0$. Therefore, On the boundary plane $g \in \{0,1 \} $ of the coupling system, the evolution trajectory of the system presents a closed periodic orbit. Specifically, when $g=0$ and $\kappa=0$, our coupled system degenerates into the coupled system described in previous work \cite{weitz2016oscillating}.
\end{proof}	

\begin{figure}[t]
	\centering 
	\includegraphics[width=0.5\textwidth]{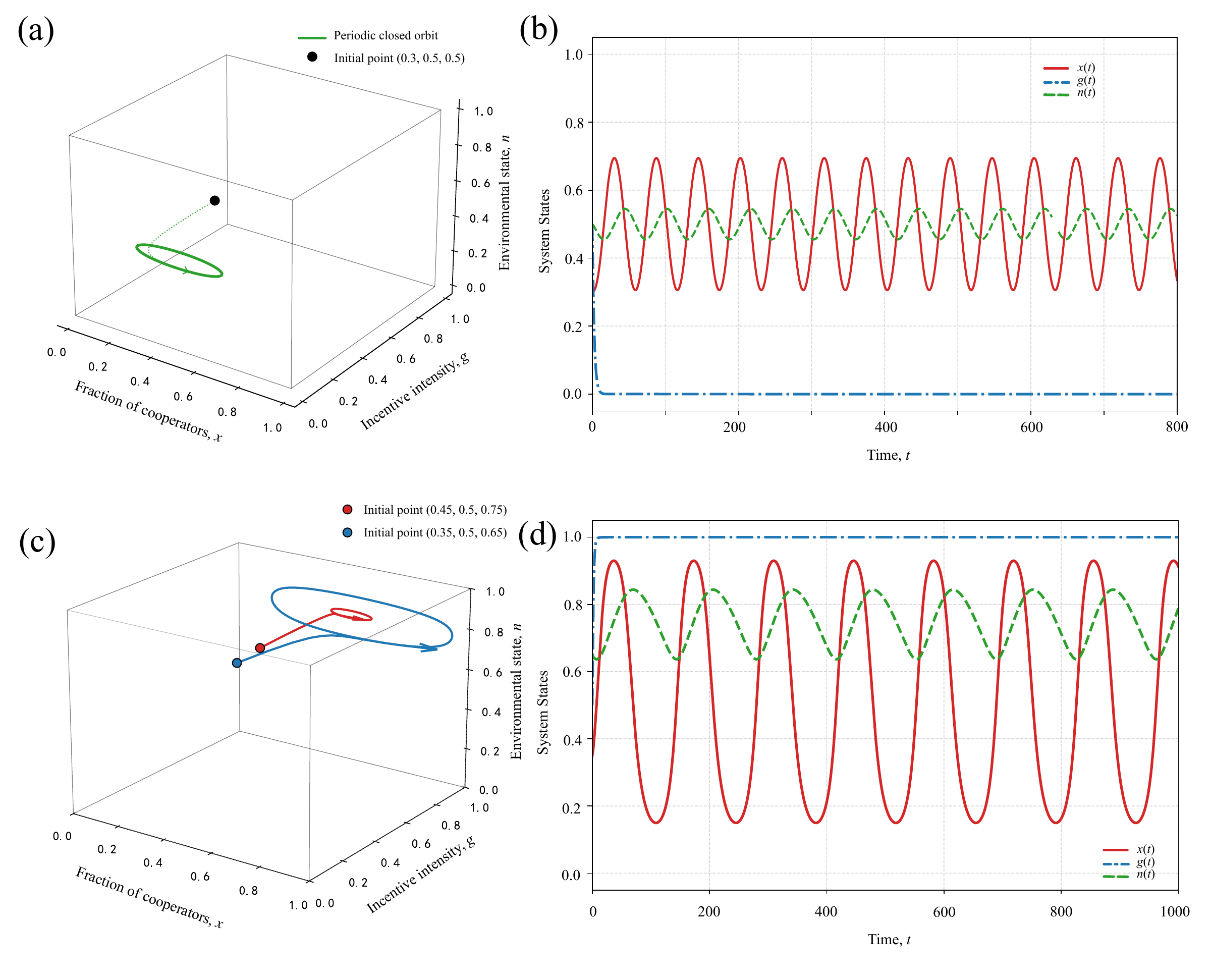}  
	\caption{\textbf{Emergence of periodic closed orbits on the boundary faces.} 3D phase portrait and corresponding frequency-time plot illustrating evolutionary trajectories converging to periodic closed orbits on the face $g=0$ in panels (a) and (b), on the face $g=1$ for panels (c) and (d). Parameters are $\Delta_{TR} = 1$, $\Delta_{PS} = 1$, $\beta_1 = 0$, $\beta_2 = 1$, $\alpha=1$, $\eta = 0.5$, $x_{target}=0$, $\delta=1$,  $\kappa = 0$, $\epsilon_1=0.05$, $\epsilon_2=0.5$, $n_{target}=0$, $\theta=1$ in panels (a) and (b); for (c-d) $\Delta_{TR} = 0.8$, $\Delta_{PS} = 0.2$, $\beta_1 = 0.1$, $\beta_2 = 2$, $\alpha=2$, $\eta = 0.4$, $x_{target}=0.9$, $\delta=0.2$,  $\kappa = 0.3$, $\epsilon_1=0.05$, $\epsilon_2=0.5$, $n_{target}=0.9$, $\theta=0.2$.
	}
	\label{fig5}  
\end{figure} 
Next, we provide numerical examples to verify the theoretical analysis results mentioned above. As shown in Figure \ref{fig4}(a), when $\kappa=0$ and $\beta_1=0$, we find periodic closed orbits on the planes of $g=0$ and $g=1$ in 3D phase space. This means that the frequency of cooperators and resource status are in a cyclic oscillation. Figures \ref{fig4}(b) and (c) respectively present the changes in system states over time, verifying the results in the phase diagram. 

In Figure \ref{fig5}, we present the numerical results of periodic closed orbits appearing on the $g=0$ and $g=1$ planes, respectively. On the $g=0$ plane, when $\kappa=0$, we can prove that the coupled system becomes a Hamiltonian system. Therefore, a closed orbit can appear (Figures \ref{fig5}(a) and (b)). On the $g=1$ plane, when $\beta_1\neq 0$, the parameters must satisfy $\frac{\kappa}{\beta_1}=\frac{\Delta_{TR}-\Delta_{PS}}{\Delta_{PS}}$, and the system will become a non dissipative conservative system. In Figures \ref{fig5}(c) and (d), we present relevant numerical validations
	
\begin{figure}[h]
	\centering 
	\includegraphics[width=0.5\textwidth]{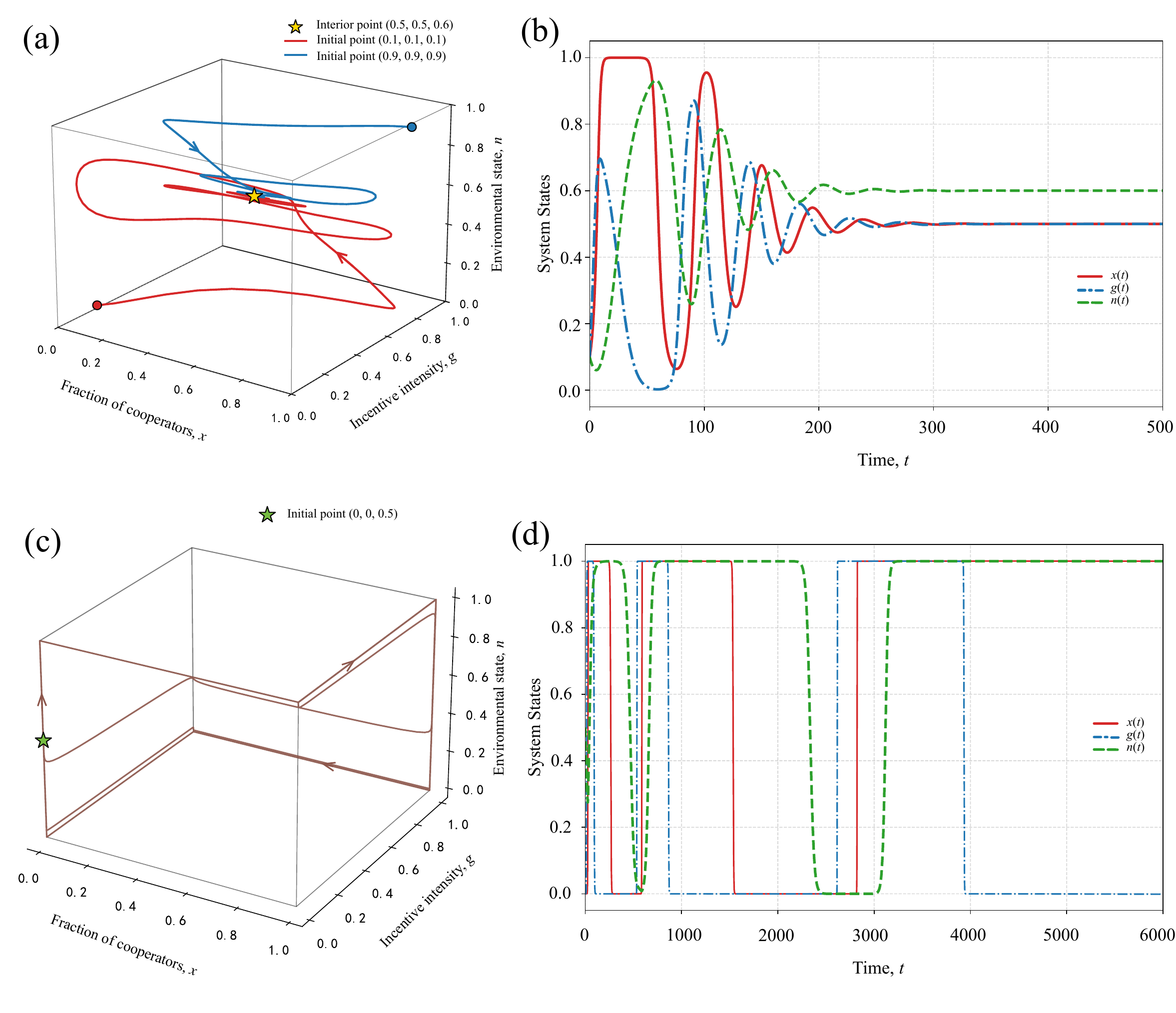}  
	\caption{\textbf{Global dynamics of the interior equilibrium and heteroclinic cycle under distinct parameter configurations.} Panels (a) and (b) show the 3D phase portrait and corresponding frequency-time plot of the interior equilibrium $(0.5, 0.5, 0.6)$, and limit cycle in panels (c) and (d). Parameters are $\Delta_{TR} = 0.8$, $\Delta_{PS} = 0.2$, $\beta_1 = 0.1$, $\beta_2 = 1$, $\alpha=1$, $\eta = 0.6$, $x_{target}=0.7$, $\delta=0.8$,  $\kappa = 0.1$, $\epsilon_1=0.2$, $\epsilon_2=0.5$, $n_{target}=0.8$, $\theta=0.4$ in panels (a) and (b); for (c-d) $\Delta_{TR} = 0.8$, $\Delta_{PS} = 0.4$, $\beta_1 = 0.6$, $\beta_2 = 1$, $\alpha=3$, $\eta = 0.5$, $x_{target}=0.2$, $\delta=0.2$,  $\kappa = 0.5$, $\epsilon_1=0.05$, $\epsilon_2=0.5$, $n_{target}=0.9$, $\theta=0.8$.
	}
	\label{fig6}  
\end{figure} 
\begin{theorem}\label{theo3.5}
For three-dimensional coupled systems, there is an internal equilibrium point $(x^*,g^*,n^*)$ if and only if the nonlinear equations
$$\begin{cases} (1-2n^*)[\Delta_{PS} + x^*(\Delta_{TR} - \Delta_{PS})] + \beta_1 g^* + \kappa x^* = 0 \\ \alpha(n_{target} - n^*) + \beta_2(x_{target} - x^*) - \delta g^* = 0 \\ (\theta + 1)x^* - 1 + \eta g^* = 0 \end{cases}$$
has at least one set of real number solutions in the open set $(0,1)^3 $. The equilibrium point is stable if the characteristic polynomial $\lambda^3 + a_1 \lambda^2 + a_2 \lambda + a_3 = 0$ of the Jacobian matrix strictly satisfies the Routh Hurwitz criterion \cite{golnaraghi2010automatic}: $$a_1>0, \quad a_3>0, \quad a_1 a_2-a_3>0,$$ where $a_1=-\text{Tr}(J)$, $a_2= (J_{11}J_{22} - J_{12}J_{21}) - J_{13}J_{31} - J_{23}J_{32}$, and $a_3 = -\det(J)$. $J_{ij}$ represents the elements of the Jacobian matrix.
\end{theorem}

\begin{proof}
	Here we provide a brief proof. By setting the right-hand side of three differential equations equal to 0, we can directly obtain the system of equations in Theorem \ref{theo3.5}. To determine the local stability of the internal equilibrium point of a nonlinear system, we calculate the Jacobian matrix $J$ of the system at that point. The characteristic equation is $\vert{}\lambda I - J_{14}\vert{} = 0$, which is expanded to obtain a third-order polynomial: $$\lambda^3 + a_1 \lambda^2 + a_2 \lambda + a_3 = 0,$$ where the coefficients are $a_1 = -(J_{11} + J_{22})$,$a_2 = J_{11}J_{22} - J_{12}J_{21} - J_{13}J_{31} - J_{23}J_{32}$,$a_3 = J_{11}J_{23}J_{32} - J_{12}J_{23}J_{31} - J_{13}(J_{21}J_{32} - J_{22}J_{31}).$ According to the Routh Hurwitz criterion, if the real part of all eigenvalues is negative, namely, $\text{Re} (\lambda_i)<0, i \in \{1,2,3 \} $, the conclusion of Theorem \ref{theo3.5} can be obtained.
\end{proof}

For the convenience of analysis, we will rewrite the coupled system as
$\dot{x} = x(1-x)F_x(x,g,n),$
$\dot{g} = g(1-g)F_g(x,g,n),$
$\dot{n} = n(1-n)F_n(x,g,n).$

\begin{theorem}\label{theo3.6}
Due to the intrinsic positive payoff $\Delta_{PS} > 0$, the boundary heteroclinic cycle sequentially connecting the six nodes $E_1(0,0,0)$, $E_3(0,1,0)$, $E_5(1,1,0)$, $E_8(1,1,1)$, $E_6(1,0,1)$, $E_4(0,0,1)$ is topologically broken at $E_1(0,0,0)$ with a positive transverse eigenvalue $\lambda_x(E_1) = \Delta_{PS} > 0$. Then the system exhibits an asymptotically stable limit cycle in the $O(\epsilon)$ neighborhood of the boundary, provided the characteristic compression ratio satisfies
$\rho = \prod_{i=1}^{6} \frac{\vert{}\lambda_{i}^{(c)}\vert{}}{\lambda_{i}^{(e)}}> 1,$
where the local expansive ($\lambda_i^{(e)}$) and contractive ($\lambda_i^{(c)}$) eigenvalues are defined as:
$$\begin{cases} \lambda_{1}^{(e)} = F_g(0,0,0) > 0, & \lambda_{1}^{(c)} = F_n(0,0,0) < 0 \\ \lambda_{2}^{(e)} = F_x(0,1,0) > 0, & \lambda_{2}^{(c)} = -F_g(0,1,0) < 0 \\ \lambda_{3}^{(e)} = F_n(1,1,0) > 0, & \lambda_{3}^{(c)} = -F_x(1,1,0) < 0 \\ \lambda_{4}^{(e)} = -F_g(1,1,1) > 0, & \lambda_{4}^{(c)} = -F_n(1,1,1) < 0 \\ \lambda_{5}^{(e)} = -F_x(1,0,1) > 0, & \lambda_{5}^{(c)} = F_g(1,0,1) < 0 \\ \lambda_{6}^{(e)} = -F_n(0,0,1) > 0, & \lambda_{6}^{(c)} = F_x(0,0,1) < 0. \end{cases}$$
\end{theorem}

\begin{proof}
We also provide a brief proof here. At $E_1(0,0,0)$, the transverse eigenvalue along the $x$-axis is $\lambda_x = \Delta_{PS} > 0$. Then the local orbital derivative near $E_1$ satisfies: $\frac{dg}{dx} \approx \frac{\epsilon_2 \lambda_g}{\lambda_x} \frac{g}{x} = K \frac{g}{x},$
where the rate ratio $K \gg 1$. Integrating yields the geometric envelope $g \propto x^K$. This means any small leak in $x$ makes $g$ increase extremely fast, instantly pulling the trajectory back to the boundary surface.
 
The overall stability is determined by the compression ratio $\rho$. Because $\rho = \prod_{i=1}^6 \frac{\vert{}\lambda_i^{(c)}\vert{}}{\lambda_i^{(e)}} > 1$, the total inward attraction is stronger than the local expansion at the nodes. This global contraction ensures that the trajectory does not diverge, but instead forms a stable, closed limit cycle very close to the boundary.
\end{proof}

In Figure \ref{fig6}, we provide numerical verification of the two theorems mentioned above. As shown in the top row of Figure \ref{fig6}, when the model parameters satisfy the conditions of Theorem \ref{theo3.5}, there exists a stable internal equilibrium point in the phase space, which means that using constant excitation can maintain a constant level of cooperation while ensuring resource sustainability. When the model parameters satisfy the conditions of Theorem 6, we find that the internal equilibrium point becomes unstable, and the system trajectory forms an oscillation dynamic along six nodes $E (0,0,0) \rightarrow E (0,1,0) \rightarrow  E (1,1,0) \rightarrow  E(1,1,1) \rightarrow  E(1,0,1) \rightarrow  E(0,0,1) \rightarrow  E(0,0,0)$ (see Figures \ref{fig6}(c)and (d)). This means that the system will still fall into the oscillating tragedy of the commons.

    \begin{figure}[h]
	\centering 
	\includegraphics[width=0.5\textwidth]{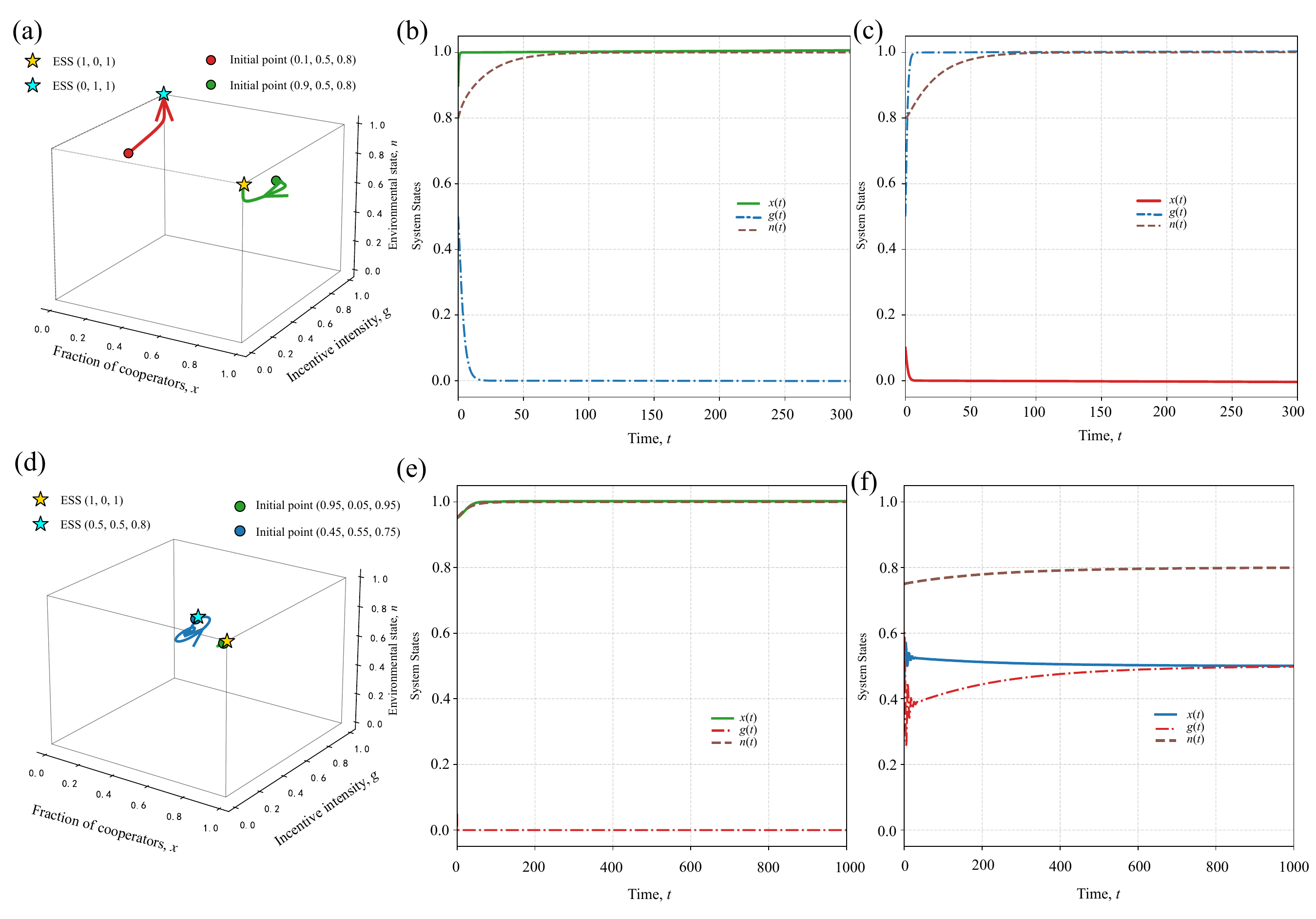}  
	\caption{\textbf{Bistability and time-series evolution in the tripartite co-evolutionary system.} Panels (a)-(c) show bistability between the autonomous equilibrium (1, 0, 1) and the regulation-dependent equilibrium (0, 1, 1). Panels (d)-(f) show coexistence of the corner equilibrium (1, 0, 1) and an interior equilibrium (0.5, 0.5, 0.8). Parameters are $\Delta_{TR} = 0.8$, $\Delta_{PS} = 1.5$, $\beta_1 = 0.1$, $\beta_2 = 3$, $\alpha=1$, $\eta = 2$, $x_{target}=0.8$, $\delta=0.5$,  $\kappa = 3$, $\epsilon_1=0.05$, $\epsilon_2=0.5$, $n_{target}=0.8$, $\theta=1$ in panels (a) - (c); $\Delta_{TR} = 0.1$, $\Delta_{PS} = 2$, $\beta_1 = 1.06$, $\beta_2 = 25$, $\alpha=1$, $\eta = 0.2$, $x_{target}=0.6$, $\delta=5$,  $\kappa = 0.2$, $\epsilon_1=0.05$, $\epsilon_2=0.5$, $n_{target}=0.8$, $\theta=0.8$ in panels (d) - (f);
	}
	\label{fig7}  
\end{figure}

Finally, our coupled system also exhibits some bistable results, such as $E (1,0,1)$ and $E(0,1,1)$; $E(1,0,1)$ and $E(x^*, g^*, n^*) $; $E(1,0,1)$ and $E(0, g^*, 1)$. Here we mainly focus on the first two bistable results. When the following inequality holds
$$\begin{cases} \kappa > \Delta_{TR} \\ \Delta_{PS} > \beta_1 \\ \eta > 1 \\  \alpha(n_{target} - 1) + \beta_2 x_{target} - \delta > 0 \end{cases}.$$
We know that $E (1,0,1)$ and $E(0,1,1)$ are both stable. Numerical experiments are presented in Figures \ref{fig7} (a)-(c)).
When $\kappa > \Delta_{TR}$, $0 < x^*, g^*, n^* < 1$, and $a_1 > 0,  a_3 > 0,a_1 a_2 - a_3 > 0$, $E(1,0,1)$ and $E(x^*, g^*, n^*)$ are both stable. Numerical experiments are presented in Figures \ref{fig7} (d)-(f)).

\section{Conclusion}
In this work, we have constructed a three-dimensional coupled dynamic system that integrates social strategies, policy interventions, and ecological environment evolution. Through theoretical analysis and numerical verification, we have revealed that the coupled system can exhibit rich dynamic phenomena such as monostability, bistability, oscillation, etc. Importantly, $E (1,0,1)$ can become the only global stable equilibrium point in phase space, which means that the system can endogenously guide multi-agent strategies towards complete cooperation ($x=1$) without relying on long-term external intervention, and achieve long-term stability of the ecological environment ($n=1 $) under zero normalized regulatory costs ($g=0$). In addition, there may be some bistable results in the system, which depict that the direction of system evolution strictly depends on the attraction domain where the initial state is located.

The interior equilibrium represents a practically attainable operating regime rather than an idealized social optimum. At this equilibrium, cooperation is sustained at an intermediate level, the resource stock remains viable, and institutional intervention persists at a positive intensity. This result indicates that sustainable governance does not necessarily require complete behavioral consensus. Instead, stability may arise from a balance among private incentives, ecological regeneration, and regulatory compensation, which suggests that institutional intervention should be
viewed as an endogenous component of system operation rather than as a temporary external correction.

The desirable state characterized by full cooperation, maximal resource availability, and minimal intervention has an important institutional interpretation. Minimal regulation should not be understood as the absence of governance. Rather, it is an endogenous consequence of a
well-functioning system in which cooperative behavior and ecological
recovery reinforce each other. Once cooperation becomes sufficiently
widespread and the resource stock approaches its target level, the
regulatory burden can be reduced without destabilizing the system.
However, premature withdrawal of institutional support may move the
system outside the attraction region of the desirable state and trigger
a return to low cooperation and resource degradation.

\appendix[Detailed theoretical analysis]

%


		In the main text, our 3D system is
		\begin{equation*}
			\label{eq:dynamical_system}
			\begin{cases} 
				\dot{x} =  x(1-x) \{(1 - 2n) \left[ x\Delta_{TR} + (1-x)\Delta_{PS} \right] + \beta_1 g + \kappa x\}, \\[8pt]
				\dot{g} = \epsilon_2 g(1-g)  \left[ \alpha(n_{target} - n) + \beta_2(x_{target} - x) - \delta g \right], \\[8pt]
				\dot{n} =\epsilon_1 n(1-n) \left[ \theta x-(1-x) + \eta g \right].
			\end{cases}
		\end{equation*}
	For the sake of convenience, we have rewritten it as
$$\begin{cases} \dot{x} = x(1-x)F_x(x, g, n) \\ \dot{g} = g(1-g)F_g(x, g, n) \\ \dot{n} = n(1-n)F_n(x, g, n). \end{cases}$$
The Jacobian matrix $J$ obtained by taking the partial derivative of the system at any point $(x, g, n)$ is as follows:
$$J =  \begin{bmatrix} J_{11} & J_{12} & J_{13} \\ J_{21} & J_{22} & J_{23} \\ J_{31} & J_{32} & J_{33} \end{bmatrix},$$ where
$J_{11} = (1 - 2x)F_x + x(1 - x)\frac{\partial F_x}{\partial x}$,$J_{12} = x(1 - x)\frac{\partial F_x}{\partial g}$,$J_{13} = x(1 - x)\frac{\partial F_x}{\partial n}$,$J_{21} = g(1 - g)\frac{\partial F_g}{\partial x}$,$J_{22} = (1 - 2g)F_g + g(1 - g)\frac{\partial F_g}{\partial g}$,$J_{23} = g(1 - g)\frac{\partial F_g}{\partial n}$,$J_{31} = n(1 - n)\frac{\partial F_n}{\partial x}$,$J_{32} = n(1 - n)\frac{\partial F_n}{\partial g}$,$J_{33} = (1 - 2n)F_n + n(1 - n)\frac{\partial F_n}{\partial n}$.

The Jacobian matrix at the equilibrium points of the eight corners of the system can be written as
$$J(0,0,0) \begin{bmatrix} \Delta_{PS} & 0 & 0 \\ 0 & \epsilon_2(\alpha n_{target} + \beta_2 x_{target}) & 0 \\ 0 & 0 & -\epsilon_1 \end{bmatrix}.$$ Then it is unstable since $\Delta_{PS}>0$.
$$J(1,0,0) =  \begin{bmatrix} -(\Delta_{TR} + \kappa) & 0 & 0 \\ 0 & J_{22} & 0 \\ 0 & 0 & \epsilon_1 \theta \end{bmatrix},$$ where $J_{22}=\epsilon_2[\alpha n_{target} + \beta_2 (x_{target} - 1)].$ Then it is unstable since $\epsilon_1 \theta>0$.
$$J(0,1,0) =  \begin{bmatrix} \Delta_{PS} + \beta_1 & 0 & 0 \\ 0 & J_{22}& 0 \\ 0 & 0 & \epsilon_1(\eta - 1) \end{bmatrix},$$ where $J_{22}=-\epsilon_2(\alpha n_{target} + \beta_2 x_{target} - \delta).$ Then it is unstable since $ \Delta_{PS} + \beta_1>0$.
$$J(0,0,1) =  \begin{bmatrix} -\Delta_{PS} & 0 & 0 \\ 0 & \epsilon_2[\alpha (n_{target} - 1) + \beta_2 x_{target}] & 0 \\ 0 & 0 & \epsilon_1 \end{bmatrix}.$$ Then it is unstable since $\epsilon_1>0$.
$$J(1,1,0) =  \begin{bmatrix} -(\Delta_{TR} + \beta_1 + \kappa) & 0 & 0 \\ 0 & J_{22} & 0 \\ 0 & 0 & \epsilon_1(\theta + \eta) \end{bmatrix},$$ where $J_{22}=-\epsilon_2[\alpha n_{target} + \beta_2 (x_{target} - 1) - \delta].$ Then it is unstable since $\epsilon_1(\theta + \eta)>0$.
$$J(1,0,1) =  \begin{bmatrix} \Delta_{TR} - \kappa & 0 & 0 \\ 0 & \epsilon_2[\alpha (n_{target} - 1) + \beta_2 (x_{target} - 1)] & 0 \\ 0 & 0 & -\epsilon_1 \theta \end{bmatrix},$$ Then it is stable when $\Delta_{TR} - \kappa<0.$

$$J(0,1,1) =  \begin{bmatrix} \beta_1 - \Delta_{PS} & 0 & 0 \\ 0 & J_{22} & 0 \\ 0 & 0 & -\epsilon_1(\eta - 1) \end{bmatrix},$$ where $J_{22}=-\epsilon_2[\alpha (n_{target} - 1) + \beta_2 x_{target} - \delta]$. It is stable when $\beta_1 - \Delta_{PS}<0$, $J_{22}<0$, and $\eta - 1>0$.
$$J(1,1,1) =  \begin{bmatrix} \Delta_{TR} - \beta_1 - \kappa & 0 & 0 \\ 0 & J_{22} & 0 \\ 0 & 0 & -\epsilon_1(\theta + \eta) \end{bmatrix},$$ where $J_{22}=-\epsilon_2[\alpha (n_{target} - 1) + \beta_2 (x_{target} - 1) - \delta]$. It is unstable since $J_{22}>0$.

Therefore, we can easily obtain the conclusion of Theorem 1 in main text.

For $E(0, \frac{\alpha (n_{target}-1)+\beta_2 x_{target}}{\delta}, 1)$,

$$J=  \begin{bmatrix} \beta_1 g^* - \Delta_{PS} & 0 & 0 \\ -\epsilon_2 \beta_2 g^*(1 - g^*) & -\epsilon_2 \delta g^*(1 - g^*) & -\epsilon_2 \alpha g^*(1 - g^*) \\ 0 & 0 & \epsilon_1 (1 - \eta g^*) \end{bmatrix},$$		
where $g^{*}=\frac{\alpha (n_{target}-1)+\beta_2 x_{target}}{\delta},$ it is stable when $\beta_1 g ^ *<\Delta_{PS}$ and $\eta g^*>1$.

For $E(\frac{\Delta_{PS}-\beta_1}{\Delta_{PS}-\Delta_{TR}+\kappa}, 1, 1)$,
$$J =  \begin{bmatrix} J_{11} & x^*(1-x^*)\beta_1 & -2x^*(1-x^*)[x^*\Delta_{TR} + (1-x^*)\Delta_{PS}] \\ 0 & J_{22}& 0 \\ 0 & 0 & -\epsilon_1[\theta x^* - (1-x^*) + \eta] \end{bmatrix},$$
where $J_{11}=x^*(1-x^*)(\Delta_{PS} - \Delta_{TR} + \kappa), J_{22}=-\epsilon_2[\alpha(n_{target} - 1) + \beta_2(x_{target} - x^*) - \delta]$, $x^{*}=\frac{\Delta_{PS}-\beta_1}{\Delta_{PS}-\Delta_{TR}+\kappa}$, it is stable when $\Delta_{PS} - \Delta_{TR} + \kappa < 0,  \delta < \alpha(n_{target} - 1) + \beta_2(x_{target} - \frac{\Delta_{PS}-\beta_1}{\Delta_{PS}-\Delta_{TR}+\kappa})$ and  $(\theta+1)\frac{\Delta_{PS}-\beta_1}{\Delta_{PS}-\Delta_{TR}+\kappa}+\eta>1$.

Therefore, we can easily obtain the conclusion of Theorem 2 in main text.

For $E(0, \frac{\alpha n_{target}+\beta_2 x_{target}}{\delta}, 0)$,
$$J =  \begin{bmatrix} \Delta_{PS} + \beta_1 g^* & 0 & 0 \\ -\epsilon_2 \beta_2 g^*(1 - g^*) & -\epsilon_2 \delta g^*(1 - g^*) & -\epsilon_2 \alpha g^*(1 - g^*) \\ 0 & 0 & \epsilon_1 (\eta g^* - 1) \end{bmatrix}.$$ It is unstable since $\Delta_{PS} + \beta_1 g^*>0$.

For $E(1, \frac{\alpha n_{target}+\beta_2 (x_{target}-1)}{\delta}, 0)$,
$$J =  \begin{bmatrix} -(\Delta_{TR} + \beta_1 g^* + \kappa) & 0 & 0 \\ -\epsilon_2 \beta_2 g^*(1 - g^*) & -\epsilon_2 \delta g^*(1 - g^*) & -\epsilon_2 \alpha g^*(1 - g^*) \\ 0 & 0 & \epsilon_1 (\theta + \eta g^*) \end{bmatrix}.$$ It is unstable since $\epsilon_1 (\theta + \eta g^*)>0$.

For $E(\frac{\Delta_{PS}}{\Delta_{PS}-\Delta_{TR}+\kappa}, 0,1)$,
$$J =  \begin{bmatrix} J_{11} & x^*(1-x^*)\beta_1 & J_{13} \\ 0 & J_{22} & 0 \\ 0 & 0 & -\epsilon_1[\theta x^* - (1-x^*)] \end{bmatrix},$$ where $J_{11}=x^*(1-x^*)(\Delta_{PS} - \Delta_{TR} + \kappa),$ $J_{22}=\epsilon_2[\alpha(n_{target} - 1) + \beta_2(x_{target} - x^*)]$, $J_{13}=-2x^*(1-x^*)[x^*\Delta_{TR} + (1-x^*)\Delta_{PS}]$. It is unstable since $\Delta_{PS}-\Delta_{TR}+\kappa>0.$

For $E(0, \frac{1}{\eta}, n_{target}+\frac{\beta_2 x_{target}-\delta/ \eta}{\alpha})$, 
$$J =  \begin{bmatrix} (1-2n^*)\Delta_{PS} + \frac{\beta_1}{\eta} & 0 & 0 \\ -\epsilon_2 \beta_2 g^*(1-g^*) & -\epsilon_2 \delta g^*(1-g^*) & -\epsilon_2 \alpha g^*(1-g^*) \\ \epsilon_1 (\theta + 1) n^*(1-n^*) & \epsilon_1 \eta n^*(1-n^*) & 0 \end{bmatrix},$$
we have $\lambda_1 = (1-2n^*)\Delta_{PS} + \frac{\beta_1}{\eta}$. For the $x=0 $plane, we have
$$J_{sub} =  \begin{bmatrix} -\epsilon_2 \delta g^*(1-g^*) & -\epsilon_2 \alpha g^*(1-g^*) \\ \epsilon_1 \eta n^*(1-n^*) & 0 \end{bmatrix}.$$
We have $Tr(J_{sub}) = -\epsilon_2 \delta g^*(1-g^*)<0$ and $Det(J_{sub}) = \epsilon_1 \epsilon_2 \alpha \eta g^*(1-g^*) n^*(1-n^*)>0$. This means that $\lambda_2$ and $\lambda_3$ are a pair of eigenvalues with a strictly negative real part. Therefore, the stable condition for this point is
$$\left[1 - 2\left(n_{target} + \frac{\beta_2 x_{target} - \delta/\eta}{\alpha}\right)\right]\Delta_{PS} + \frac{\beta_1}{\eta} < 0.$$

For $E(\frac{1-\eta}{1+\theta}, 1, 1/2+\frac{\beta_1 +\kappa \frac{1-\eta}{1+\theta}}{2[\frac{1-\eta}{1+\theta} \Delta_{TR}+(1-\frac{1-\eta}{1+\theta})\Delta_{PS}]})$,

$$J =  \begin{bmatrix} J_{11} & x^*(1-x^*)\beta_1 & J_{13} \\ 0 & J_{22} & 0 \\ \epsilon_1 n^*(1-n^*)(\theta + 1) & \epsilon_1 \eta n^*(1-n^*) & 0 \end{bmatrix},$$
where $J_{11}=x^*(1-x^*)[(1-2n^*)(\Delta_{TR}-\Delta_{PS}) + \kappa]$, $J_{22}=-\epsilon_2[\alpha(n_{target} - n^*) + \beta_2(x_{target} - x^*) - \delta],$ and $J_{13}=-2x^*(1-x^*)[x^*\Delta_{TR} + (1-x^*)\Delta_{PS}].$ We have 
$\lambda_2 = -\epsilon_2[\alpha(n_{target} - n^*) + \beta_2(x_{target} - x^*) - \delta].$
For the $g=1$ plane, we have
$$J_{sub} =  \begin{bmatrix} J_{11} & J_{12} \\ \epsilon_1 n^*(1-n^*)(\theta + 1) & 0 \end{bmatrix},$$ where $J_{11}=x^*(1-x^*)[(1-2n^*)(\Delta_{TR}-\Delta_{PS}) + \kappa].$ and $J_{12}=-2x^*(1-x^*)[x^*\Delta_{TR} + (1-x^*)\Delta_{PS}]$. Then
$Det(J_{sub}) = 2\epsilon_1 x^*(1-x^*)n^*(1-n^*)(\theta + 1)[x^*\Delta_{TR} + (1-x^*)\Delta_{PS}]>0$ and $Tr(J_{sub}) = x^*(1-x^*)[(1-2n^*)(\Delta_{TR}-\Delta_{PS}) + \kappa]$.
Accordingly, the stability of this equilibrium point requires the following conditions to be met
\begin{equation*}
	\begin{cases}
		\alpha(n_{target} -1/2+\frac{\beta_1 +\kappa \frac{1-\eta}{1+\theta}}{2[\frac{1-\eta}{1+\theta} \Delta_{TR}+(1-\frac{1-\eta}{1+\theta})\Delta_{PS}]}) \\+ \beta_2(x_{target} - \frac{1-\eta}{1+\theta}) > \delta \\
		\frac{1-\eta}{1+\theta}\Delta_{TR} + (1-\frac{1-\eta}{1+\theta})\Delta_{PS} > 0 \\
		-\frac{\beta_1 +\kappa \frac{1-\eta}{1+\theta}}{\frac{1-\eta}{1+\theta} \Delta_{TR}+(1-\frac{1-\eta}{1+\theta})\Delta_{PS}}(\Delta_{TR} - \Delta_{PS}) + \kappa < 0.
	\end{cases}
\end{equation*}

For $E(x_1^*, g_1^*, 1)$, we have
$$J =  \begin{bmatrix} J_{11} & x_1^*(1-x_1^*)\beta_1 & J_{13} \\ -\epsilon_2 \beta_2 g_1^*(1-g_1^*) & -\epsilon_2 \delta g_1^*(1-g_1^*) & -\epsilon_2 \alpha g_1^*(1-g_1^*) \\ 0 & 0 & J_{33}\end{bmatrix},$$
where $J_{11}=x_1^*(1-x_1^*)(\Delta_{PS} - \Delta_{TR} + \kappa)$, $J_{13}=-2x_1^*(1-x_1^*)[x_1^*\Delta_{TR} + (1-x_1^*)\Delta_{PS}]$, $J_{33}=-\epsilon_1[\theta x_1^* - (1-x_1^*) + \eta g_1^*].$ We have $\lambda_3 = -\epsilon_1[\theta x_1^* - (1-x_1^*) + \eta g_1^*].$
For the $n=1$ plane, we have
$$J_{sub} =  \begin{bmatrix} x_1^*(1-x_1^*)(\Delta_{PS} - \Delta_{TR} + \kappa) & x_1^*(1-x_1^*)\beta_1 \\ -\epsilon_2 \beta_2 g_1^*(1-g_1^*) & -\epsilon_2 \delta g_1^*(1-g_1^*) \end{bmatrix}.$$ Then we have $Det(J_{sub}) = \epsilon_2 x_1^*(1-x_1^*)g_1^*(1-g_1^*) [\beta_1 \beta_2 - \delta(\Delta_{PS} - \Delta_{TR} + \kappa)]$ and $Tr(J_{sub}) = x_1^*(1-x_1^*)(\Delta_{PS} - \Delta_{TR} + \kappa) - \epsilon_2 \delta g_1^*(1-g_1^*)$. The stability of this equilibrium point requires the following conditions to be met
\begin{equation*}
	\begin{cases}
		\theta x_1^* - (1-x_1^*) + \eta g_1^* > 0 \\
		\delta(\kappa - \Delta_{TR} + \Delta_{PS}) - \beta_1 \beta_2 < 0 \\
		\frac{1}{\epsilon_1} x_1^*(1-x_1^*)(\Delta_{PS} - \Delta_{TR} + \kappa) - \frac{\delta}{\epsilon_2} g_1^*(1-g_1^*) < 0.
	\end{cases}
\end{equation*}

For $E(\frac{1}{1+\theta}, 0, \frac{1}{2}+\frac{\kappa \frac{1}{1+\theta}}{2[\frac{1}{1+\theta} \Delta_{TR}+(1-\frac{1}{1+\theta})\Delta_{PS}]})$,
$$J =  \begin{bmatrix} J_{11} & x^*(1-x^*)\beta_1 & J_{13}\\ 0 & J_{22} & 0 \\ \epsilon_1 n^*(1-n^*)(\theta + 1) & \epsilon_1 \eta n^*(1-n^*) & 0 \end{bmatrix},$$
where $J_{11}=x^*(1-x^*)[(1-2n^*)(\Delta_{TR}-\Delta_{PS}) + \kappa]$, $J_{13}=-2x^*(1-x^*)[x^*\Delta_{TR} + (1-x^*)\Delta_{PS}] $, $J_{22}=\epsilon_2[\alpha(n_{target} - n^*) + \beta_2(x_{target} - x^*)].$ We have $\lambda_2 = \epsilon_2[\alpha(n_{target} - n^*) + \beta_2(x_{target} - x^*)]$.
For the $g=0$ plane, we have
$$J_{sub} =  \begin{bmatrix} J_{11} & J_{12} \\ \epsilon_1 n^*(1-n^*)(\theta + 1) & 0 \end{bmatrix},$$
where $J_{11}=x^*(1-x^*)[(1-2n^*)(\Delta_{TR}-\Delta_{PS}) + \kappa]$ and $J_{12}=-2x^*(1-x^*)[x^*\Delta_{TR} + (1-x^*)\Delta_{PS}].$ Then we have $Det(J_{sub}) = 2\epsilon_1 x^*(1-x^*)n^*(1-n^*)(\theta + 1)[x^*\Delta_{TR} + (1-x^*)\Delta_{PS}]$ and $Tr(J_{sub}) = x^*(1-x^*)[(1-2n^*)(\Delta_{TR}-\Delta_{PS}) + \kappa] = x^*(1-x^*) \frac{\kappa \Delta_{PS}}{x^*\Delta_{TR} + (1-x^*)\Delta_{PS}}>0$. Accordingly, it is unstable.

Therefore, we can easily obtain the conclusion of Theorem 3 in main text.

For $E(x^*,g^*,n^*)$, 
$$J^* =  \begin{bmatrix} J_{11} & x^*(1-x^*)\beta_1 & J_{13} \\ -\epsilon_2 g^*(1-g^*)\beta_2 & -\epsilon_2 g^*(1-g^*)\delta & -\epsilon_2 g^*(1-g^*)\alpha \\ \epsilon_1 n^*(1-n^*)(\theta + 1) & \epsilon_1 n^*(1-n^*)\eta & 0 \end{bmatrix},$$
where $J_{11}=x^*(1-x^*)[(1-2n^*)(\Delta_{TR} - \Delta_{PS}) + \kappa]$ and $J_{13}=-2x^*(1-x^*)[x^*\Delta_{TR} + (1-x^*)\Delta_{PS}].$ Then $Tr(J^*) =  x^*(1-x^*)[(1-2n^*)(\Delta_{TR} - \Delta_{PS}) + \kappa] - \epsilon_2 g^*(1-g^*)\delta$ and $Det(J^*) = J_{13}(J_{21}J_{32} - J_{22}J_{31}) - J_{23}(J_{11}J_{32} - J_{12}J_{31})$. According to Routh Hurwitz, we can obtain the stability conditions for the internal equilibrium point.

\bibliographystyle{IEEEtran}
\bibliography{Bibliography}

\end{document}